\documentclass[conference]{IEEEtran}

\usepackage{hyperref}
\usepackage{url}
\usepackage{amsmath, amsfonts, amssymb, amsthm}
\usepackage{color}
\usepackage{braket, dsfont, mathtools, bibentry}
\usepackage{textcomp}
\usepackage{enumitem}
\usepackage[normalem]{ulem}
\usepackage{cite, caption}
\usepackage{amsmath,amssymb,amsfonts}
\usepackage{mathtools, braket, lipsum, bibentry}
\usepackage{graphicx}
\usepackage{textcomp}
\usepackage[usenames,dvipsnames]{xcolor}
\usepackage[utf8]{inputenc} 
\usepackage[T1]{fontenc}    
\usepackage{booktabs}       
\usepackage{nicefrac}       
\usepackage{microtype}      
\usepackage{float}
\usepackage{enumitem}
\usepackage[caption=false,font=footnotesize]{subfig}

\def\BibTeX{{\rm B\kern-.05em{\sc i\kern-.025em b}\kern-.08em
    T\kern-.1667em\lower.7ex\hbox{E}\kern-.125emX}}
\newcommand{\toronto}{ibmq\_toronto~}

\newtheorem{theorem}{Theorem}[section]
\newtheorem{lemma}[theorem]{Lemma}

\begin{document}
\title{Characterizing the Reproducibility of Noisy Quantum Circuits}

\author{
\IEEEauthorblockN{Samudra~Dasgupta $^{1, 2, \dagger}$ and Travis S.~Humble$^{1, 2, \ddagger}$}
\IEEEauthorblockA{\textit{$^{1}$ Quantum Science Center, Oak Ridge National Laboratory, Oak Ridge, TN 37830, USA}\\
\textit{$^{2}$ Bredesen Center, University of Tennessee, Knoxville, TN 37996, USA}\\
}}

\maketitle

\begin{abstract}
The ability of a quantum computer to reproduce or replicate the results of a quantum circuit is a key concern for verifying and validating applications of quantum computing. Statistical variations in circuit outcomes that arise from ill-characterized fluctuations in device noise may lead to computational errors and irreproducible results. While device characterization offers a direct assessment of noise, an outstanding concern is how such metrics bound the reproducibility of a given quantum circuit. Here, we first directly assess the reproducibility of a noisy quantum circuit, in terms of the Hellinger distance between the computational results, and then we show that device characterization offers an analytic bound on the observed variability. We validate the method using an ensemble of single qubit test circuits, executed on a superconducting transmon processor with well-characterized readout and gate error rates. The resulting description for circuit reproducibility, in terms of a composite device parameter, is confirmed to define an upper bound on the observed Hellinger distance, across the variable test circuits. This predictive correlation between circuit outcomes and device characterization offers an efficient method for assessing the reproducibility of noisy quantum circuits.
\end{abstract}


\section{Introduction}
Quantum computing leverages phenomena such as superposition and entanglement to offer novel capabilities relative to conventional computing, such as super-polynomial reductions, in time to solution and energy consumption, as well as memory storage \cite{humble2019quantum}. Quantum advantages may be found for solving a variety of problems, such as simulating quantum many-body systems \cite{browaeys2020many}, 
solving unstructured optimization problems \cite{montanaro2016quantum}, 
achieving complex linear algebra computations \cite{harrow2009quantum}, efficiently sampling high-dimensional probability distributions \cite{shang2015monte}, and enhancing the security of communication networks \cite{espitia2021role}. 
These assessments are based on an idealized quantum computing model \cite{mermin2007quantum, nielsen2002quantum, rieffel2011quantum}, comprising of an $n$ qubit register that encodes a $2^{n}$-dimensional Hilbert space $\mathbb{C}^{2\otimes n}$. 
After initializing the register to the state $\ket{\psi} = \ket{0}^{\otimes n}$, a sequence of logical unitary transformations $\mathcal{U}$, known as quantum gates, transform the state as $\ket{\psi} \rightarrow \mathcal{U}_{k} \cdots \mathcal{U}_1 \mathcal{U}_0\ket{\psi}$. 
Ultimately, a measurement operation reads out an $n$-bit string \textit{s}, where $s \in \{0, 1\}^{\otimes n}$ with $\textrm{Pr}(s) = |\braket{s|\psi}|^2$ has the probability to observe the outcome. 
\par 
In practice, on-going efforts to realize a quantum computer introduces many additional physical processes that complicate the operational description above. For example, a crude approximation to an actual device can be conceptualized as a stack of interacting layers \cite{8123664,8638598,ball2021quantum} that includes physical qubits, a physical control layer, a hardware-aware compiler, a logical control layer (including fault tolerant quantum error correction protocols), and a logical compiler and circuit optimizer, as well as the quantum algorithms and applications. Each of these intermediate processes introduce the possibility for noise and errors that make modeling the system more complicated. Subsequent certification that an actual physical system is performing as expected represents a leading concern for validating experimental results \cite{kliesch2021theory}. This is made more difficult by the inherent randomness that often manifests in the computed results, inability to pin-point exactly where in a circuit an error has occurred, curse of dimensionality, and the inability to step through program execution \cite{carrasco2021theoretical, arute2019, ferracin2019accrediting}.
\par
While quantum computers using fault-tolerant operations may eventually alleviate some of these concerns  \cite{gottesman1998theory}, current noisy intermediate-scale quantum (NISQ \cite{preskill2019quantum}) computers are influenced heavily by noise and errors in physical operations \cite{divincenzo2000physical}. Computational errors arise due to noise in the register, imperfect gate implementations, and faulty thermodynamic control systems \cite{blume2020modeling, blume2010optimal, temme2017error, kandala2019error, wilson2020just}. 
For example, noise in the register may arise from spontaneous decay, decoherence, coupling to the environment, cross-talk, and leakage. Errors in gate operations also arise when one of the control operations goes awry, due to lack of precision in the applied pulses, e.g., pulse attenuation and distortion, pulse jitter, and frequency drift. Thermodynamic control systems that maintain the stability of the operating environment may also be a source of noise, due to limits on the cooling power of the dilution refrigerators, contaminants in the vacuum chamber, or fluctuations in the electromagnetic shielding and vibration suppression systems.
\par
The varying environment and fluctuating controls present in current NISQ computing platforms lead to transient sources of noise that impact the ability to reproduce the results of a quantum circuit  \cite{dasgupta2021stability,zhang2021predicting}. While errors arising from fixed sources of device noise can frequently be mitigated using tailored methods \cite{maciejewski2020mitigation,hamilton2020error,bravyi2021mitigating,hamilton2020scalable}, ill-characterized and transient noise impedes the reproduction of NISQ circuit results and prevents the replication of quantum computed solutions. It is, therefore, important to assess the conditions under which a noisy quantum circuit may be expected to be reproducible, as well as the correlation with the corresponding noise in the quantum device. Bounding statistical variations, expected from a noisy quantum circuit, in terms of device noise itself, is potentially an efficient method for assessing reproducibility on NISQ platforms.
\par 
Here, we show how device characterization may be used to bound the reproducibility of a noisy quantum circuit. 
First, we formalize the problem of quantifying the reproducibility of results from a quantum circuit ($\mathcal{C}$), executed on a noisy quantum device ($D$). 
We express this statistical variation, in terms of the Hellinger distance between the observed noisy output and expected ideal distribution, and we examine the conditions under which this distance deviates.
Then, we derive an upper bound on the Hellinger distance by setting a threshold on these deviations that may be expressed in terms of the device noise itself. This yields a test as to whether the quantum circuit ($\mathcal{C}$) is reproducible when executed on the given device ($D$). We show that this test is efficient when cast in terms of a single composite parameters computed by estimating the device noise. We then validate this method using an ensemble of single qubit test circuits executed on a well-characterized, but noisy, superconducting transmon processor. \color{black} 
\par
The remainder of the presentation is as follows: in Section 2, 
we show how the composite parameter ($\gamma_D(\mathcal{C})$), representing a quantum circuit ($\mathcal{C}$), executed on a device ($D$), is composed from available device characterization, and that, when $\gamma_D(\mathcal{C})$ stays below the threshold $\gamma_{\textrm{max}}(\mathcal{C})$, then the circuit is reproducible. In Section 3, 
we present results from applying these methods to single qubit circuits, executed on a superconducting transmon processor, to validate agreement between the theory and experiment. Finally, in Section 4, 
we discuss the success of this approach, as well as its potential limitations.
\par
This work is a significant extension of an earlier publication \cite{dasgupta2021reproducibility}, in which the methods and results have been revised and refined to validate experimentally-observed bounds.

\section{Method}\label{sec:method}
Consider $\rho_{0}$ to denote the initial density matrix representing the state of the $n$ qubit quantum register. In the ideal setting, the quantum state of the register evolves under a unitary transformation $\mathds{U}$ that describes the circuit $\mathcal{C}$ as: 
\begin{equation}
\rho = \mathds{U}_{\mathcal{C}} \rho_{0} \mathds{U}_{\mathcal{C}}^\dagger.
\end{equation}
Let $\{\ket{i}\}_{i=0}^{N-1}$ represent the orthonormal computational basis, with $N=2^n$, and let $\{\Pi_i = \ket{i}\bra{i}\}_{i=0}^{N-1}$ be the corresponding set of orthonormal projectors. The corresponding probability distribution for the results generated by an ideal quantum device ($D^{ideal}$) is then given as:
\begin{equation}
\mathds{P}^{ideal} = \{p_i^{ideal}\} \textrm{ for } i \in \{0, 1, \cdots, N-1\}
\end{equation}
where $p_i^{ideal} = \textrm{Tr}[\Pi_i \rho]$ is the probability to observe the $i$-th outcome. {In general, it is not efficient to construct the set $\mathds{P}^{ideal}$ using conventional methods, i.e., classical computing, and the resources for such calculations scale exponentially in the size of $n$. However, such demanding calculations are feasible for values of $n < 50$, while calculations of highly structured circuits, such as the quantum Fourier transform, quantum search, or arithmetic circuits, may be simulated efficiently. We limit our current approach to the consideration of $\mathcal{C}$ for such instances.
}
\par
In the presence of noise, the evolution of the quantum register is no longer modeled by a unitary operator, and the final state is generally a mixed-state. Let $\mathcal{E}$ be the super-operator representing a noisy quantum channel, in terms of Kraus operators ($M_k$) that represent the execution of the circuit $\mathcal{C}$. The corresponding state of the register is then: 
\begin{equation}
\rho' = \mathcal{E}(\rho)= \sum\limits_{k} M_k \rho_{0} M_k^\dagger
\end{equation}
where the Kraus operators are a function of the device reliability parameters. Thus, the corresponding probability distribution from executing $\mathcal{C}$ on the noisy quantum device $D^{noisy}$ is given by:
\begin{equation}
\mathds{P}^{noisy} = \{p_i^{noisy}\} \textrm{ for } i \in \{0, 1, \cdots, N-1\}
\end{equation}
where $p_i^{noisy} = \textrm{Tr}[M_i^\dagger M_i \rho']$, and $M_i$ is the measurement operator for a noisy readout channel \cite{smith2021qubit}. 
\par
We quantify the variation between the ideal and actual distributions for the results from executing the circuit $\mathcal{C}$ using the Hellinger distance. We denote the Hellinger distance between $\mathds{P}^{ideal}$ and $\mathds{P}^{noisy}$ as:
\begin{equation}
d(\mathds{P}^{\textrm{ideal}}, \mathds{P}^{\textrm{noisy}}) = 
\sqrt{1-BC(\mathds{P}^{\textrm{ideal}},\mathds{P}^{\textrm{noisy}})}
\label{eq:hellinger}
\end{equation}
with the Bhattacharyya coefficient $BC(\mathds{P}^{\textrm{ideal}}, \mathds{P}^{\textrm{noisy}}) \in [0,1]$ defined as:
\begin{equation}
BC(\mathds{P}^{\textrm{ideal}}, \mathds{P}^{\textrm{noisy}}) = 
\sum\limits_{i=0}^{N-1} \sqrt{p_i^{\textrm{ideal}} p_i^{\textrm{noisy}}}.
\end{equation}
Note that the Hellinger distance vanishes for identical distributions and approaches unity for those distributions with no overlap.
Now, let $\delta \in [0,1]$ be a parameter that defines a threshold for the Hellinger distance, i.e., when:
\begin{equation}
d(\mathds{P}^{\textrm{ideal}}, \mathds{P}^{\textrm{noisy}}) \leq \delta
\label{eq:repr}
\end{equation}
then, the results from a quantum circuit ($\mathcal{C}$) on a noisy, quantum device ($D$) are reproducible within a distance ($\delta$) of the idealized outcomes. 
The equation for characterizing the reproducibility of a noisy quantum circuit introduces a test for deciding if the observed value of the Hellinger distance lies below the elected threshold. 
\par 
In general, the above test for reproducibility requires estimating the noisy probability distribution ($\mathds{P}^{\textrm{noisy}}$) by repeatedly sampling the circuit outcomes. {The number of samples required to estimate an arbitrary distribution over $N=2^n$ outcomes with precision ($\epsilon$) scales exponentially with the Hilbert space dimension ($n$) and as an inverse square of the required precision ($\epsilon$) (see Appendix~\ref{sec:appendix_precision}).}
Such a direct method for estimating the Hellinger distance is inefficient and requires repetition across devices. We next consider how to correlate the bound placed on the Hellinger distance in \mbox{Equation (\ref{eq:repr})}, with the noise assumptions for the given device. We show, by example, that a composite parameter ($\gamma_{D}(\mathcal{C})$), characterizing the noisy quantum circuit, can be similarly bound from above as: 
\begin{equation}
\gamma_{D}(\mathcal{C}) \leq \gamma_{\textrm{max}}(\mathcal{C}, \delta)
\label{eq:gamma_tau}.
\end{equation}
\subsection{Example}
Consider an $n$ qubit state prepared as a uniform superposition across the $2^n$ computational basis states as:
\begin{equation}
\ket{\psi} = \frac{1}{2^{n/2}}\sum\limits_{s \in \{0,1\}^n}\ket{s}.
\end{equation}
The corresponding quantum circuit shown in \mbox{Figure \ref{fig:27_hadamards_with_n_1}} corresponds to:
\begin{equation}
\ket{\psi} = \mathds{H}^{\otimes n} \ket{0}^{\otimes n}
\end{equation}
for which the ideal distribution is:
\begin{equation}
p_i^{\textrm{ideal}} = \frac{1}{N} = 2^{-n}.
\end{equation}
We next assume that the quantum register is reliably initialized as $\ket{0}^{\otimes n}$ and that inter-qubit cross-talk is negligible during noisy circuit execution. Rather, gate error and readout fidelity capture the principal sources of noise in this circuit, each of which is modeled as a noisy process.
\begin{figure}[H]
\center
\includegraphics[width=8cm]{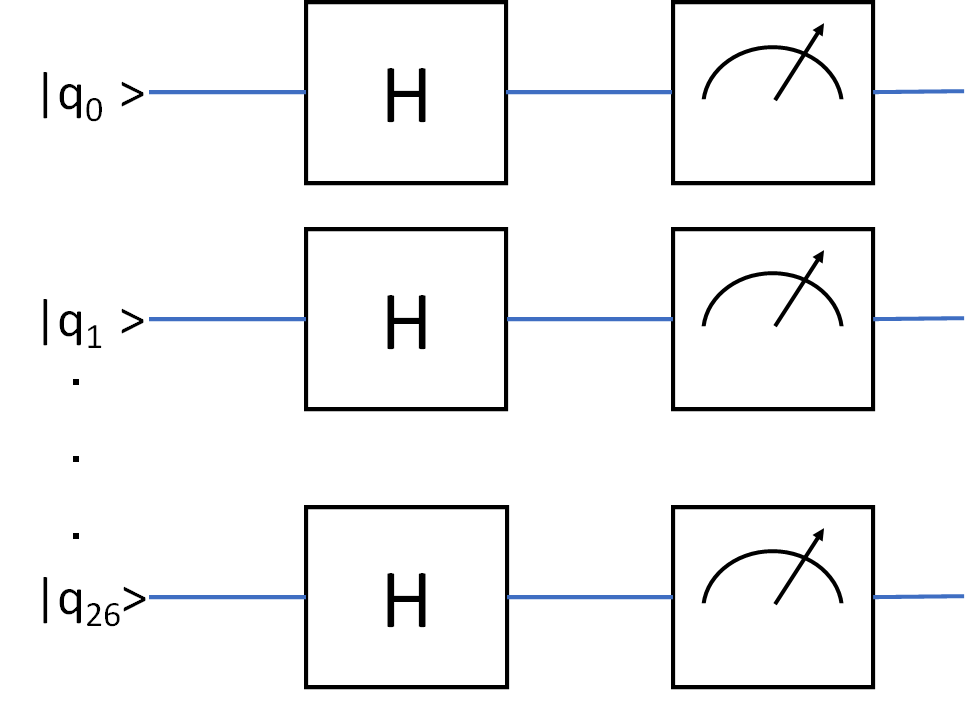} 
\caption{Circuit used for our experiment. In this figure, $H$ represents the Hadamard gate. The meter symbol denotes measurement gate.}
\label{fig:27_hadamards_with_n_1}
\end{figure}
\par 
Let $\mathds{I}, \mathds{X}, \mathds{Y},$ and  $\mathds{Z}$ denote the $2\times 2$ identity matrix, Pauli-$X$ matrix, Pauli-$Y$ matrix, and the Pauli-$Z$ matrix, respectively, in the computational basis as:
\begin{equation}
\begin{split}
\mathds{I}=&\begin{pmatrix}
1& 0\\
0& 1\\
\end{pmatrix}\\
\mathds{X}=&\begin{pmatrix}
0& 1\\
1& 0\\
\end{pmatrix}\\
\mathds{Y}=&\begin{pmatrix}
0& -i\\
i& 0\\
\end{pmatrix}\\
\mathds{Z}=&\begin{pmatrix}
1&  0\\
0& -1\\
\end{pmatrix}\\
\end{split}
\end{equation}
Additionally, let $R_Y(\alpha)$ denote the rotation by an angle ($\alpha$) about the Y-axis on the Bloch sphere:
\begin{equation}
\begin{split}
R_Y(\alpha) =& e^{-i\frac{\alpha}{2}Y}\\
=&\cos \frac{\alpha}{2} \mathds{I} - i \sin \frac{\alpha}{2} \mathds{Y}\\
=&\begin{pmatrix}
\cos\frac{\alpha}{2}  & -\sin\frac{\alpha}{2} \\
\sin\frac{\alpha}{2} & \cos\frac{\alpha}{2} \\
\end{pmatrix}\\
\end{split}.
\end{equation}
An ideal Hadamard gate is then given by:
\begin{equation}
H = R_Y\left( \frac{\pi}{2} \right) \mathds{Z} =
\begin{pmatrix}
\cos\frac{\pi}{4} &  \sin\frac{\pi}{4} \\
\sin\frac{\pi}{4} & -\cos\frac{\pi}{4} \\
\end{pmatrix}
= \frac{1}{\sqrt{2}}
\begin{pmatrix}
1&1\\
1&-1\\
\end{pmatrix}.
\end{equation}
We model a noisy Hadamard gate ($\tilde{\mathds{H}}$) by the unitary: 
\begin{equation}
\tilde{\mathds{H}} = 
\begin{pmatrix}
\cos\left(\frac{\pi}{4}+\theta\right) &  \sin\left(\frac{\pi}{4}+\theta\right) \\
\sin\left(\frac{\pi}{4}+\theta\right) & -\cos\left(\frac{\pi}{4}+\theta\right) \\
\end{pmatrix}
\end{equation}
\color{black}
where $\theta$ is the implementation error that is assumed to be small, i.e., $\theta \ll 1$.
\par
The operator representation for a unitary control error has only one term which can be seen as follows. 
Write a noisy unitary $\tilde{\mathcal{U}}$ as 
\begin{equation}
    \tilde{\mathcal{U}} = \tilde{\mathcal{U}} \mathcal{U}^\dagger \mathcal{U}
\end{equation}
where $\mathcal{U}$ is the ideal unitary. Thus, 
\begin{equation}
\begin{split}
\rho^\prime =& \tilde{\mathcal{U}} \rho  \tilde{\mathcal{U}}^\dagger \\
=& (\tilde{\mathcal{U}} \mathcal{U}^\dagger)  (\mathcal{U} \rho \mathcal{U}^\dagger) (\mathcal{U}  \tilde{\mathcal{U}}^\dagger) \\
=& E (\mathcal{U} \rho \mathcal{U}^\dagger)  E^\dagger \\
\end{split}
\end{equation}
where $E$ is the operator representing the unitary control error that arises due to imperfections in the control system. For our example,
\begin{equation}
    E = \tilde{\mathds{H}}\mathds{H}^\dagger = 
\begin{pmatrix}
\cos\theta &  -\sin\theta \\
\sin\theta & \cos\theta \\
\end{pmatrix}
\end{equation}
which is the 2D rotation matrix (no error will correspond to an identity channel with $E=\mathds{I}$).
\par 
Thus,
\begin{equation}
\tilde{\mathds{H}} =\frac{1}{\sqrt{2}}
\begin{pmatrix}
\cos\theta - \sin\theta &  \cos\theta + \sin\theta \\
\cos\theta + \sin\theta & -\cos\theta + \sin\theta \\
\end{pmatrix}
\end{equation}
\begin{equation}
\tilde{\mathds{H}}\ket{0} = 
\frac{1}{\sqrt{2}}(\cos\theta - \sin\theta)\ket{0} + 
\frac{1}{\sqrt{2}}(\cos\theta + \sin\theta)\ket{1}
\end{equation}
In the absence of readout noise, when we initialize a qubit in the ground state, subject it to a noisy Hadamard gate, and measure in the $\mathds{Z}$-basis, we get the probabilities for observing the output as:
\begin{equation}
\textrm{Pr}(0) = \frac{1}{2}(1-\sin2\theta)
\end{equation}
\begin{equation}
\textrm{Pr}(1) = \frac{1}{2}(1+\sin2\theta)
\end{equation}
\par
We next consider what happens when the Hadamard gate is followed by a noisy measurement. The readout channel is characterized as a quantum channel,  \cite{oszmaniec2019simulating, bravyi2021mitigating, geller2020rigorous} as shown in \mbox{Figure  \ref{fig:f0f1}}, using two parameters for each qubit. 

\begin{figure}[H]
\centering
\includegraphics[width=8cm]{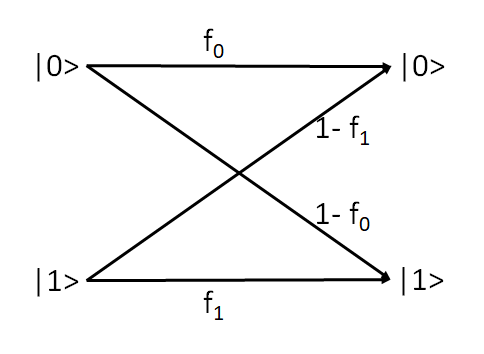}
\caption{The quantum channel maps $\ket{0}$ and $\ket{1}$ to their expected outcomes, with probabilities $f_0$ and $f_1$, respectively.}
\label{fig:f0f1}
\end{figure}

The first parameter ($f_0$) defines the probability of observing $0$ post readout when the channel input state is $\ket{0}$, and the second ($f_1$) defines the probability of observing $1$ post readout when the channel input state is $\ket{1}$.

{A classical representation for the single qubit readout channel is:
\begin{equation}
\mathds{P}_{\textrm{obs}} = \Lambda \mathds{P}_{\textrm{true}}    
\end{equation}
where $\Lambda$ is the readout error matrix with elements $\Lambda_{ij}$ = probability of observing $\ket{i}$ when the input to channel is $\ket{j} (i, j \in \{0, 1\})$:
\begin{equation}
\Lambda =  
\begin{pmatrix}
f_0 &  1-f_1 \\
f_1 & 1-f_0 \\
\end{pmatrix}
\label{eq:Lambda}    
\end{equation}
Equivalently, the quantum channel representation for a single qubit noisy measurement has two terms, i.e., $M_0$ and $M_1$, and is specified by a super-operator ($\mathcal{E}$), whose action on the quantum state is as follows \cite{smith2021qubit}:
\begin{equation}
\begin{split}
\mathcal{E}(\rho) =& M_0 \rho M_0^\dagger + M_1 \rho M_1^\dagger\\
M_0 =& \sqrt{f_0}\ket{0}\bra{0}+\sqrt{1-f_1}\ket{1}\bra{1}\\
M_1 =& \sqrt{1-f_0}\ket{0}\bra{0}+\sqrt{f_1}\ket{1}\bra{1}\\
\end{split}
\end{equation}
which is equivalent to \mbox{Equation (\ref{eq:Lambda})}, when you consider the action of the measurement operators ($M_i$), given by $Pr(i)_{\textrm{obs}}= \textrm{Tr}\{M_i^\dagger M_i \rho\}$. Thus, the representation equivalently shows how a noisy measurement of a quantum state $\rho$ yields a particular readout $i$. 
}


Let Pr$(0)$ be the probability of observing 0 when we prepare a qubit in the ground state, subject it to a Hadamard gate, and then measure it. Let Pr$(1)$ be the corresponding probability of observing 1 for the same experiment (i.e., we prepare a qubit in the ground state, subject it to a Hadamard gate, and then measure it). Additionally, let $f$ be the average readout fidelity and $\epsilon$ be the readout fidelity asymmetry:
\begin{equation}
f = \frac{f_0+f_1}{2}
\label{eq:f_def}
\end{equation}
\begin{equation}
\epsilon = f_0 - f_1
\label{eq:eps_def}.
\end{equation}
Thus, in the presence of readout noise, the probability of observing 0 and 1 for each qubit in \mbox{Figure \ref{fig:27_hadamards_with_n_1}} is given by \cite{smith2021qubit}:
\begin{equation}
\begin{split}
\textrm{Pr}(0) =& {\textrm{Tr}\{M_0^\dagger M_0 \mathcal{E}(\rho)\}}\\
=& \frac{1}{2}(1-\sin2\theta)f_0 + \frac{1}{2}(1+\sin2\theta)(1-f_1)\\
=& \frac{1 + \epsilon - 2\sin2\theta(f-\frac{1}{2})}{2}\\
=& \frac{1+\gamma}{2}
\end{split}
\label{eq:pr0}
\end{equation}
\begin{equation}
\begin{split}
\textrm{Pr}(1) =& {\textrm{Tr}\{M_1^\dagger M_1 \mathcal{E}(\rho)\}}\\
=& \frac{1}{2}(1-\sin2\theta)(1-f_0) + \frac{1}{2}(1+\sin2\theta)f_1\\
=& \frac{1 - \epsilon + 2\sin2\theta(f-\frac{1}{2})}{2}\\
=& \frac{1-\gamma}{2}
\end{split}
\end{equation}
where
\begin{equation}
\gamma = \epsilon - 2\sin2\theta\left(f-\frac{1}{2}\right)
\label{eq:theta_def}.
\end{equation}

\par 
Let $(s_{n-1} s_{n-2} \cdots s_0)$ represent the $n$-bit string with $s_i \in \{0,1\}$, and let $s = \sum\limits_{i=0}^{n-1}2^i s_i$ be the decimal integer equivalent. In the absence of cross-talk between gates, 
$\mathds{P}^{noisy} = \{p_s^{noisy}\}_{s=0}^{N-1}$ where:
\begin{equation}
p_s^{noisy}
=\prod\limits_{i=0}^{n-1} 
\left( \frac{1+\gamma_i}{2} \right)^{1-s_i}
\left( \frac{1-\gamma_i}{2} \right)^{s_i}\\
\end{equation}
and $\gamma_i$ refers to the $i$-th register element. Thus,
\begin{equation}
\begin{split}
BC( \mathds{P}^{ideal}, \mathds{P}^{noisy}) =&
\sum\limits_{s=0}^{N-1}\sqrt{\frac{1}{2^n} \prod\limits_{i=0}^{n-1}
\left( \frac{1+\gamma_i}{2} \right)^{1-s_i}
\left( \frac{1-\gamma_i}{2} \right)^{s_i}
}\\
=& \frac{1}{2^n} \sum\limits_{s=0}^{N-1} \prod\limits_{i=0}^{n-1} \sqrt{ 1+\gamma_i }^{1-s_i}
\sqrt{ 1-\gamma_i}^{s_i}\\
\end{split}
\end{equation}
and the sum is over all bit strings $\tilde{s}=(s_{n-1} s_{n-2} \cdots s_0)$ with decimal integer equivalent $s = \sum\limits_{i=0}^{n-1}2^i s_i$. If $\delta$ is the user-defined threshold on the acceptable Hellinger distance, then it follows that:
\begin{equation}
1-BC( \mathds{P}^{ideal}, \mathds{P}^{noisy}) \leq \delta^2
\end{equation}
or
\begin{equation}
\frac{1}{2^n}
\sum\limits_{s=0}^{N-1}
\prod\limits_{i=0}^{n-1}
\sqrt{ 1+\gamma_i }^{1-s_i}
\sqrt{ 1-\gamma_i}^{s_i}
\leq \delta^2
\label{eq:hadamard_gamma}.
\end{equation}
Without loss of generality, we consider that all register elements have identical readout  and gate errors, which allows us to replace $\gamma_i$ with $\gamma$ for each of the $n$ register elements. Thence:
\begin{equation}
\begin{split}
1- \frac{1}{2^n} \sum\limits_{s=0}^{N-1} \prod\limits_{i=0}^{n-1} \sqrt{ 1+\gamma_i }^{1-s_i} \sqrt{ 1-\gamma_i}^{s_i} &\leq \delta^2\\
\Rightarrow 
1- \frac{1}{2^n} \sum\limits_{s=0}^{N-1} \prod\limits_{i=0}^{n-1} \sqrt{ 1+\gamma}^{1-s_i} \sqrt{ 1-\gamma}^{s_i} &\leq \delta^2\\
1- \frac{1}{2^n} \sum\limits_{k=0}^{n-1} {n \choose k} \sqrt{ 1+\gamma}^{k} \sqrt{ 1-\gamma}^{n-k} &\leq \delta^2\\
1- \frac{1}{2^n} \left( \sqrt{ 1+\gamma} + \sqrt{ 1-\gamma} \right)^n &\leq \delta^2\\
1- \left( \frac{\sqrt{ 1+\gamma} + \sqrt{ 1-\gamma}}{2} \right)^n &\leq \delta^2\\
\label{eq:simple}
\end{split}.
\end{equation}
As shown in Appendix B, when assuming that $\delta$ is small, this yields:
\color{black}
\begin{equation}
\gamma_D(\mathcal{C}) \leq \gamma_{\textrm{max}}(\mathcal{C})
\label{eq:repeated_ineq}
\end{equation}
with
\begin{equation}
\gamma_{\textrm{max}}(\mathcal{C}) = 2(1-\delta^2)^{1/n} \sqrt{1-(1-\delta^2)^{2/n}}
\end{equation}
and
\begin{equation}
\gamma_D(\mathcal{C}) = \left| \epsilon - 2\sin2\theta\left(f-\frac{1}{2}\right)\right|.
\end{equation}
{Here, $\gamma_D(\mathcal{C})$ is a composite parameter, estimated using characterization of the individual gate operations, e.g., during periodic device calibration. For our specific circuit and device noise model, the parameter is a function of the gate angle error and measurement fidelity and serves as a combined measure of the hardware noise and output variation; the exact relationship is a function of the noisy circuit model used.} Thus, we have shown that if $\gamma_D(\mathcal{C})$ exceeds the bound $\gamma_{\textrm{max}}(\mathcal{C})$, then the threshold on the Hellinger distance is also exceeded. Notably, the former statement does not require experimental execution of the circuit itself or even estimation of the Hellinger distance, but rather depends solely on our model for the noisy device and accuracy of the parameter estimation.
\par
As we have assumed, $\delta$ is small, and \mbox{Equation (\ref{eq:simple})} may be reduced as:
\begin{equation}
\delta \geq \frac{1}{2}\sqrt{\frac{n}{2}}\gamma_D(\mathcal{C})
\end{equation}
which tells us that the tolerance specification for an experimentally-observed output distribution when using a noisy device ($D$) that is lower bounded. Attempts to reproduce the experiment must use an error bar larger than this minimum, otherwise it will likely fail the accuracy test, and {the ability of the device to reproduce the output distribution will be questioned.}
\section{Results}\label{sec:results}
We validated our method for testing the reproducibility of the previously described test circuits on the superconducting transmon hardware, called \toronto, from IBM, whose schematic is shown in \mbox{Figure \ref{fig:toronto}}. The test circuit shown in \mbox{Figure \ref{fig:27_hadamards_with_n_1}} was programmed using the IBM qiskit toolkit \cite{ibm_quantum_experience_website} and compiled and executed remotely on the \toronto device on 8 April 2021. 
\par
To estimate the device parameters, we repeated our experiments {to estimate the readout fidelity and gate angle error} $L$ times (each). Let $l$ denote the index of the $l$-th experiment. For any instance of circuit execution, the \toronto device prepared an ensemble of $S$ identical circuits, where $S$ denotes the number of shots and $s$ denotes the $s$-th shot in a particular experiment. In the tests reported below, $L=203$ was the number of repetitions successfully executed during a 30-min reservation-window on \toronto, and $S=8,192$ was the number of shots, the maximum allowed by the device. We separately analyzed the results for the case $n = 1$ in \mbox{Equation (\ref{eq:repeated_ineq})} using each of the 27 register elements available in \toronto.
\par
Next, we characterized the device parameters required to estimate  $\gamma_D(\mathcal{C})$  and $\gamma_{\textrm{max}}$. Here, we introduce the convention that a random variable is denoted by bold font, and a caret sign denotes a particular realization of that random variable. We first characterized readout fidelity, in which SPAM(0) denotes an experiment with a register element, prepared as $\ket{0}$ and measured. Similarly, SPAM(1) denotes an experiment in with a register element, prepared as $\ket{1}$ measured. Then, let $\boldsymbol{b}^{SPAM(0)}_{l,s,q}$ denote the binary outcome of measuring in the computational basis, when collecting the $s$-th shot of the $l$-th experiment of the SPAM(0) circuit on the $q$-th register element. Additionally, let $\boldsymbol{f}_1^q(l)$ denote the initialization fidelity observed in the $l$-th SPAM(1) experiment for the $q$-th register element. Similarly, let $\boldsymbol{f}_0^q(l)$ denote the initialization fidelity observed in the $l$-th SPAM(0) experiment for the $q$-th register element. Thus:
\begin{equation}
\boldsymbol{f}_1^q(l) = \frac{\sum\limits_{s=1}^{S} \boldsymbol{b}^{SPAM(1)}_{l,s,q}}{S}
\end{equation}

\begin{equation}
\hat{f}_1^q(l) = \frac{\sum\limits_{s=1}^{S} \hat{b}^{SPAM(1)}_{l,s,q}}{S}
\label{eq:spam1}
\end{equation}

and

\begin{equation}
\boldsymbol{f}_0^q(l) = 1-\frac{\sum\limits_{s=1}^{S} \boldsymbol{b}^{SPAM(0)}_{l,s,q}}{S}
\end{equation}
\begin{equation}
\hat{f}_0^q(l) = 1-\frac{\sum\limits_{s=1}^{S} \hat{b}^{SPAM(0)}_{l,s,q}}{S}
\label{eq:spam0}.
\end{equation}

Let $\boldsymbol{\epsilon}^q_l$ denote the realized fidelity asymmetry of the $q$-th register element in the $l$-th experiment. Thus: 
\begin{equation}
\boldsymbol{\epsilon}^q_l = \boldsymbol{f}_0^q(l) - \boldsymbol{f}_1^q(l).
\end{equation}
Let $\boldsymbol{\bar{\epsilon}}_q$ denote the mean of the fidelity asymmetry for the $q$-th register element over the $L$ experiments, and let $\hat{\bar{\epsilon}}_q$ be the corresponding observed value.
Thus:
\begin{equation}
\boldsymbol{\bar{\epsilon}}^q = \frac{\sum\limits_{l=1}^{L} \boldsymbol{\epsilon}_l^q }{L}
\label{eq:epsilon_bar}
\end{equation}
and
\begin{equation}
\hat{\bar{\epsilon}}_q = \frac{\sum\limits_{l=1}^{L} \hat{\epsilon}_l^q}{L}.
\end{equation}

To quantify the error on these measurements, we define $\boldsymbol{\sigma}( \boldsymbol{\bar{\epsilon}}^q )$ as the standard deviation of population mean $\boldsymbol{\bar{\epsilon}}^q$, such that:
\begin{equation}
\hat{\sigma}^2( \boldsymbol{\bar{\epsilon}}^q ) =  \frac{\hat{\sigma}^2( \boldsymbol{\epsilon}^q )}{L}
= \frac{1}{L(L-1)} \sum\limits_{l=1}^{L} \left( \hat{\epsilon}^q_l - \hat{\bar{\epsilon}}^q_l \right)^2.
\end{equation}
The average readout fidelity $f^q$ for each qubit $q$ is then calculated using \mbox{Equation (\ref{eq:f_def})}. 
\par
\mbox{Figure \ref{fig:f0f1_toronto_qubit_0_onwards_spruce_2021}} shows the experimental results for the fidelity distributions of $f_0$ and $f_1$ for the first nine register elements of \toronto, collected on 8 April 2021, between 8:00-10:00pm (UTC-05:00).  The initialization fidelities of the computational states are not the same and, for some register elements, they are very distinct. The asymmetric nature of the single qubit noise channel is brought out starkly by the negligible overlap between the distributions of $f_0$ and $f_1$ for qubit~$5$. Additionally, observe the significant spread in values (indicating channel instability) in qubit 3, relative to the others. These results show that the naive approach of assuming a single value for readout error for a qubit is dangerous. Not only do we have to characterize $f_0$ and $f_1$ separately, our work must also take into account the significant dispersion around the mean. \color{black} The remaining register elements are shown in the appendix in \mbox{Figures \ref{fig:f0f1_toronto_qubit_9_onwards_spruce_2021}} and \ref{fig:f0f1_toronto_qubit_18_onwards_spruce_2021}. The register variation of the readout asymmetry is shown in \mbox{Figure~\ref{fig:asymmetry_toronto_04082021_spruce}}.

\begin{figure}[H]
  \centering
  \includegraphics[width=8cm]{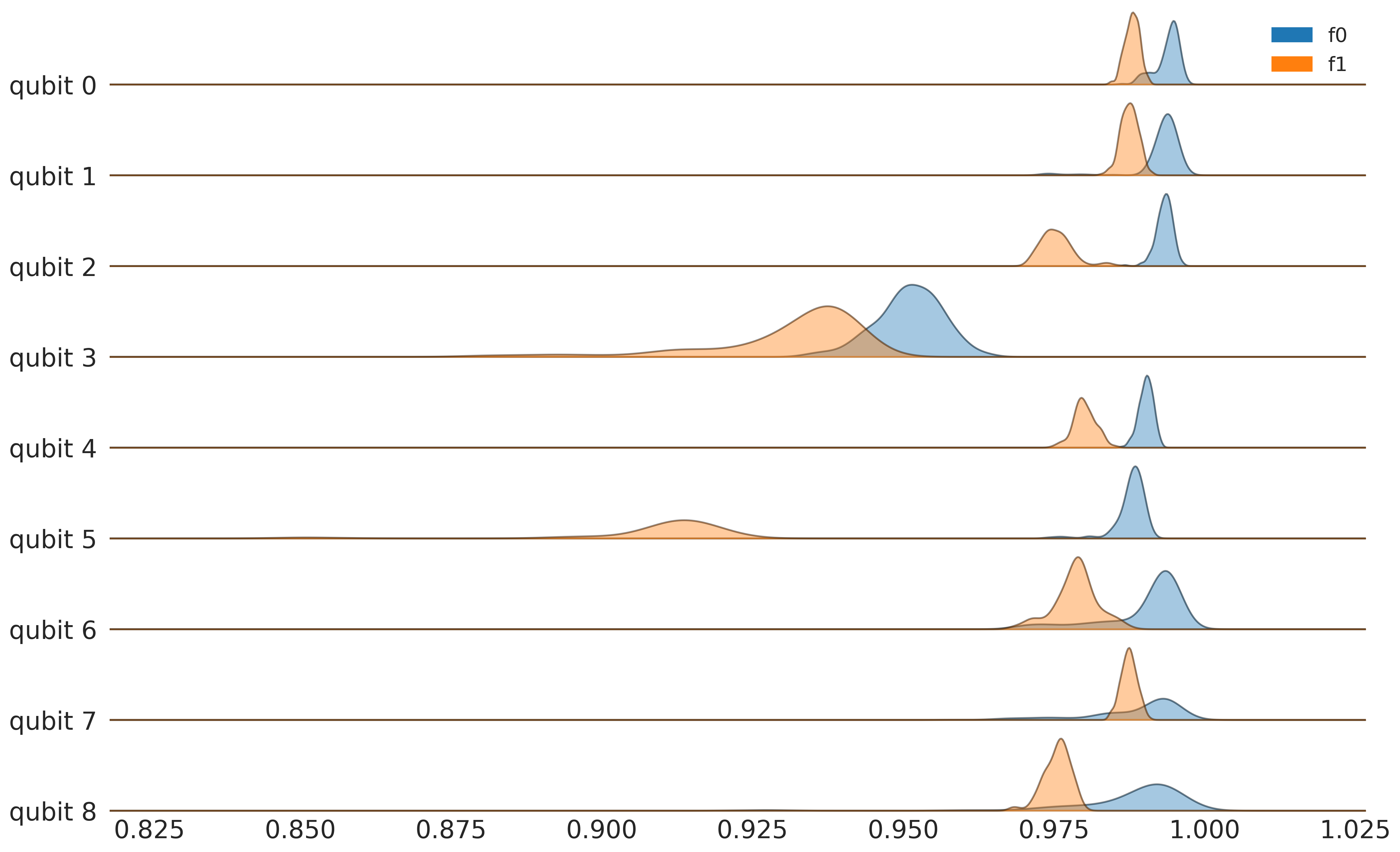}%
  \caption{Fidelity distributions for computational basis states of \toronto for qubits $0-8$. Raw data for \toronto, collected on 8 April 8 2021, between 8:00-10:00pm (UTC-05:00). The remaining register elements are shown in the appendix in \mbox{Figures~\ref{fig:f0f1_toronto_qubit_9_onwards_spruce_2021}} and \ref{fig:f0f1_toronto_qubit_18_onwards_spruce_2021}.}
  \label{fig:f0f1_toronto_qubit_0_onwards_spruce_2021}
\end{figure}

The probability (Pr$^q(0)$) for each qubit, as defined by \mbox{Equation (\ref{eq:pr0})}, was estimated by executing the circuit $\mathcal{C}$ in \mbox{Figure \ref{fig:27_hadamards_with_n_1}} and counting the faction of zeros in the 8192-bit long binary string, returned by the remote server. Let $\boldsymbol{b}_{l,s,q}^{\mathcal{C}}$ denote the random measurement outcome (a classical bit) when we conduct a $\mathcal{C}$ experiment and measure the $q$-th register element in the $l$-th experiment's $s$-th shot (measurement done in the computational $Z$-basis). Let $\hat{b}_{l,s,q}^{\mathcal{C}}$ denote the corresponding observed value. Similarly, let $\boldsymbol{\textrm{Pr}}_l^q(1)$ denote the probability of observing 1 as the outcome in the $l$-th experiment for the $q$-th register element. Similarly, let $\boldsymbol{\textrm{Pr}}_l^q(0)$ denote the probability of observing 0 as the outcome in the $l$-th experiment for the $q$-th register element. Thus:
\begin{equation}
\boldsymbol{\textrm{Pr}}_l^q(1) = \frac{\sum\limits_{s=1}^{S} \boldsymbol{b}_{l,s,q}^{\mathcal{C}}}{S}
\end{equation}
\begin{equation}
\hat{\textrm{Pr}}_l^q(1) = \frac{\sum\limits_{s=1}^{S} \hat{b}_{l,s,q}^{\mathcal{C}}}{S}
\end{equation}
and
\begin{equation}
\boldsymbol{\textrm{Pr}}_l^q(0) = 1 - \boldsymbol{\textrm{Pr}}_l^q(1)
\end{equation}
\begin{equation}
\hat{\textrm{Pr}}_l^q(0) = 1 - \hat{\textrm{Pr}}_l^q(1).
\end{equation}
\par
Let $\boldsymbol{d}_l^q$ denote the Hellinger distance between the noisy and ideal outcomes in the $l$-th experiment for the $q$-th register element. Let $\hat{d}_l^q$ be the corresponding observed value (or realization). Let $\boldsymbol{\bar{d}}^q$ denote the mean (a random variable) of the Hellinger distance for the $q$-th register element over $L$ experiments. Let $\hat{\bar{d}}^q$ be the corresponding observed value (or realization). Thus:
\begin{equation}
\boldsymbol{d}_l^q = \sqrt{ 1- \sqrt{\frac{1}{2} \boldsymbol{\textrm{Pr}}_l^q(0)} - \sqrt{\frac{1}{2} \boldsymbol{\textrm{Pr}}_l^q(1)} }
\end{equation}
\begin{equation}
\hat{d}_l^q = \sqrt{ 1- \sqrt{\frac{1}{2} \hat{\textrm{Pr}}_l^q(0)} - \sqrt{\frac{1}{2} \hat{\textrm{Pr}}_l^q(1)  } }
\end{equation}
\begin{equation}
\boldsymbol{\bar{d}}^q = \frac{\sum\limits_{l=1}^{L} \boldsymbol{d}_l^q }{L}
\end{equation}
\begin{equation}
\hat{\bar{d}}^q = \frac{\sum\limits_{l=1}^{L} \hat{d}_l^q }{L}
\end{equation}

To quantify the error on these measurements, define $
\boldsymbol{\sigma}( \boldsymbol{\bar{d}}^q )$ as the standard deviation of population mean $\boldsymbol{\bar{d}}^q$. Thus:

\begin{equation}
\hat{\sigma}^2( \boldsymbol{\bar{d}}^q ) =  \frac{\hat{\sigma}^2( \boldsymbol{d}^q )}{L}
= \frac{1}{L(L-1)} \sum\limits_{l=1}^{L} \left( \hat{d}^q_l - \hat{\bar{d}}^q \right)^2.
\end{equation}
The register variation of the experimentally-obtained Hellinger distance is shown in \mbox{Figure \ref{fig:hellinger_toronto_04082021_spruce.png}}. The Hadamard gate angle error was subsequently estimated using \mbox{Equation (\ref{eq:theta_def})}. The register variation of the the Hadamard gate angle error (in degrees) is shown in \mbox{Figure \ref{fig:theta_sq_qubitwise}}. {The dotted red line denotes the register mean for the gate error, averaged over all qubits. The gate implementation on qubit $21$ is the closest to ideal, while qubit $24$ is the farthest.}
\begin{figure}[H]
  \centering
  \includegraphics[width=8cm]{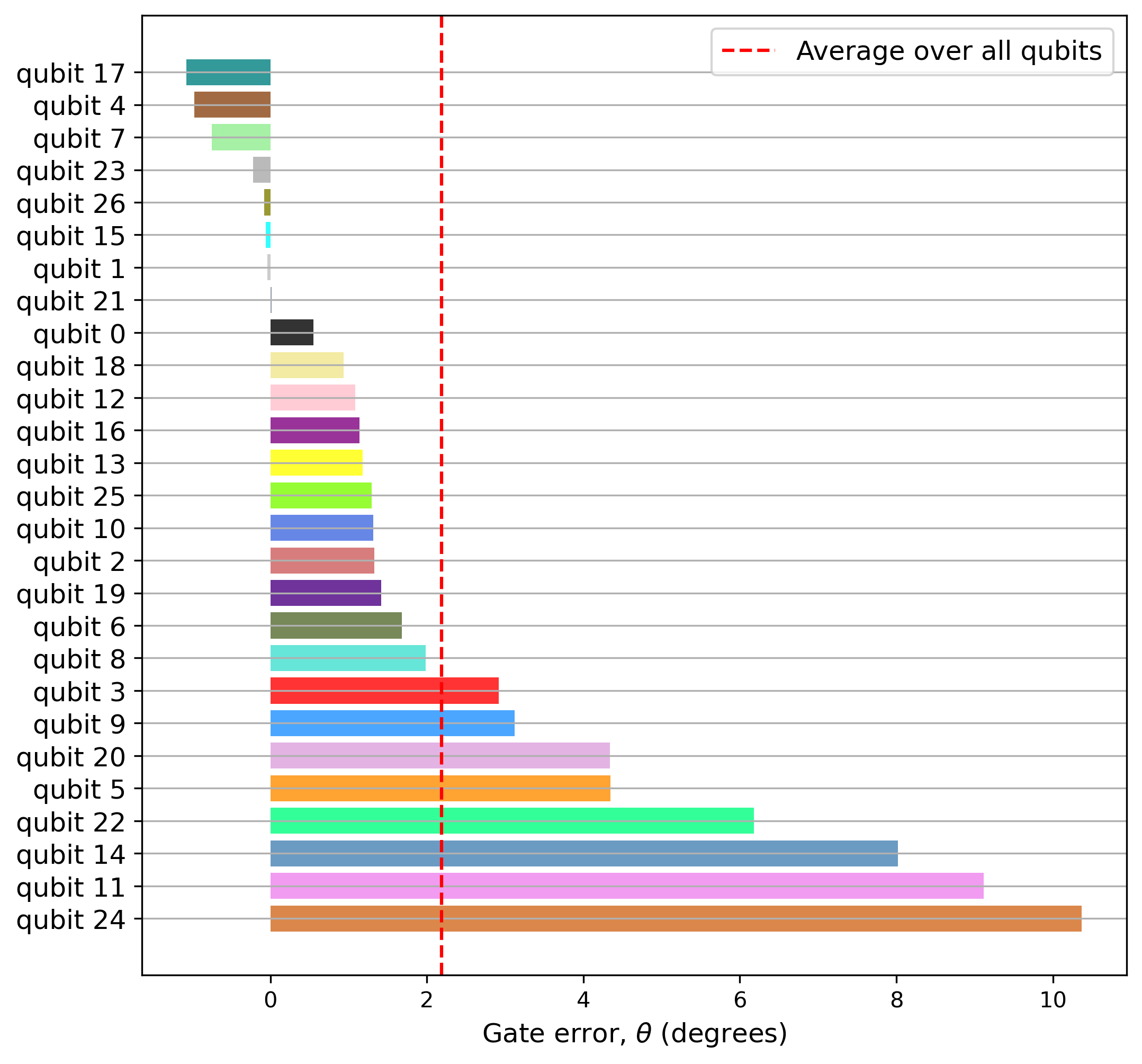}
  \captionof{figure}{Register variation of the experimentally-obtained Hadamard gate error (in degrees). Each of the 27 register elements was used to verify \mbox{Equation (\ref{eq:repeated_ineq})} for $n=1$. The dotted red line denotes the register mean for the gate error (averaged over all qubits). Qubit $21$ is the closest to ideal, while qubit $24$ is the farthest. A consistent register color scheme has been used for all the figures.}
  \label{fig:theta_sq_qubitwise}
\end{figure}

\mbox{Figure \ref{fig:model_verification_corrected_lower}} shows the values for $\gamma_{\textrm{max}}$ and $\gamma_D$ for \toronto on 8 April 2021, when $\delta$ is set to be the observed Hellinger distance. The blue dots are experimentally-observed data for each register element (see Table~\ref{tab:gamma_vals} for the full list), using the characterization data versus the actual observed Hellinger distance. It validates our noise model as \mbox{Equation (\ref{eq:gamma_tau})} holds. {The dashed line in \mbox{Figure \ref{fig:model_verification_corrected_lower}} provides the decision boundary to test circuit reproducibility, using characterization data as a proxy. Given a user-defined threshold on the statistical distance between the observed distribution and reference to be reproduced, the plot provides an upper bound for the device parameter $\gamma_D$ (the register variation of $\gamma_D$ is shown in \mbox{Figure \ref{fig:gamma_qubitwise}}). The latter must lie below this boundary for reproducibility by the device.} We conjecture that the magnitude of $|\gamma_{\textrm{max}}-\gamma_D|$ serves as a reliability indicator, i.e., higher value provides greater cushion against temporal fluctuations (see \cite{dasgupta2021stability}). Thus, Table~\ref{tab:gamma_vals} can serve as a basis for register selection tailored to a specific unitary when channel characterization data is available.

\begin{figure}[H]
  \centering
  \includegraphics[width=8cm]{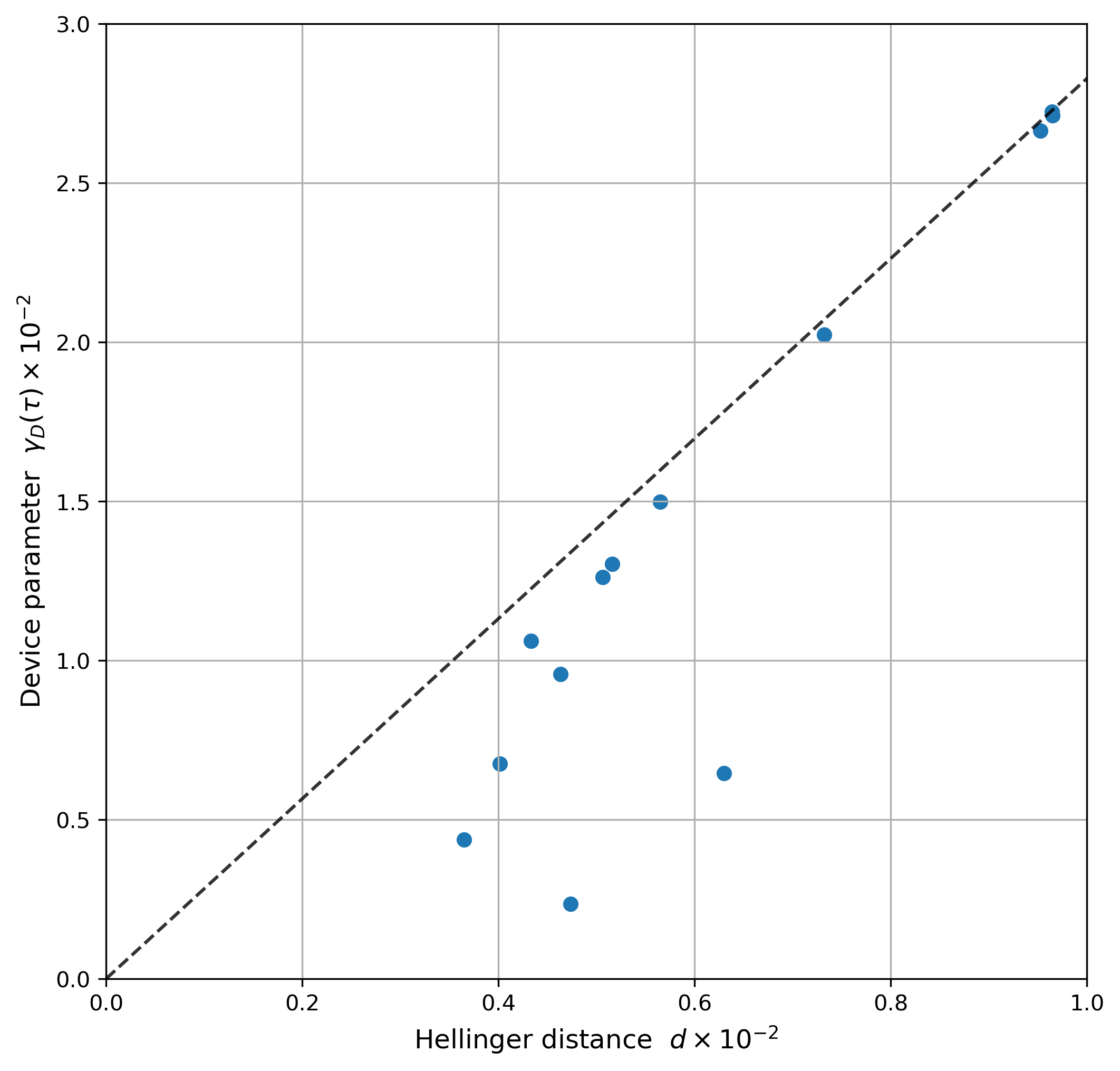}
  \captionof{figure}{Characterizing circuit reproducibility on \toronto. Plot of $\gamma_{\textrm{max}}$ (dashed line) and $\gamma_D$ (blue dots) for \toronto on 8 April 2021, when $\delta$ is set to be the observed Hellinger distance $(d)$.  \mbox{Equation (\ref{eq:gamma_tau})} should hold for all $\delta$ and hence it must hold for the corner case when $\delta$ is set to be the experimentally observed Hellinger distance. The blue dots are experimentally-observed data plotted using the characterization data versus the actual observed Hellinger distance $(d)$ for each register element. Only a subset of qubits are shown. See Table~\ref{tab:gamma_vals} for all 27 qubits.}
  \label{fig:model_verification_corrected_lower}
\end{figure}%

\section{Conclusion}\label{sec:conclusion}
Reproducibility is important for validating the performance of applications in quantum computing and as a measure of consistency in computation. Current NISQ devices are strongly affected by intrinsic noise that lead to a variety of computational error mechanisms. Here, we have characterized the ability to reproduce the output generated by noisy quantum circuit, using single qubit rotation gates and asymmetric noisy readout. We showed how to derive a composite parameter that can be easily computed from the available device characterization data (e.g., readout fidelity and gate error), such that, if the derived parameter stays below a defined threshold, then the circuit is assured of reproducible output. Our model led to an analytic bound, in terms of the Hellinger distance between computational outputs, and we validated the resulting test using experiments on a publicly available transmon processor. 
\par 
Despite the successful validation, we note that this approach has its limitations. A preliminary version of these results used only readout asymmetry and ignored the Hadamard gate error for describing circuit noise \cite{dasgupta2021reproducibility}. For example, on Aug 21, 2021, the typical mean error rates were $(7.38 \pm 5.90) \times 10^{-4}$ for single-qubit Pauli-X error and $(1.15 \pm 1.08) \times 10^{-1}$ for readout assignment error on the $27$-qubit IBM quantum  device \toronto. See \mbox{Figure \ref{fig:scatter_data}} for the results from that first-order model that failed the test of Lemma A1 for 5 of the 27 qubits, due to incorrect assumptions about the device model. An analytically-derived composite parameter, using a wrong noise model, cannot serve as a suitable proxy for circuit reproducibility. Similarly, if the device parameters themselves cannot be estimated from the available device characterization data, then the scheme is impractical.
\par
While our method is generalizable in principle, a closed form expression for $\gamma_D(\mathcal{C})$ may not be obtainable for complex circuits. As the number of parameters needed to characterize the channel increases with circuit complexity (i.e., number of qubits and circuit depth); one will need to either utilize numerical methods to deduce $\gamma_D(\mathcal{C})$ or make simplifying assumptions. A circuit-specific modular or layered characterization for the device may reduce the parameter estimation overhead, but this remains to be investigated. There is still an open question regarding how to systematically extend this method to any model and what conditions must the model meet.
\par
As quantum computations enter the realm of advantage over classical, we expect that the task of practically characterizing circuit reproducibility will become an important area. To our knowledge, this is a first attempt to frame the problem of linking circuit reproducibility to device reliability. Such reproducibility characterization serves four purposes. Firstly, it shows how to make use of the device characterization data in a targeted way for inferring circuit reproducibility. Second, it provides a basis for selecting a device for a circuit prior to execution. Third, it provides insight on how the quality of final digital output is related to the characteristics of the intermediate quantum channels and, hence, understand the device improvement pathways for a specific circuit. Lastly, it helps to estimate a lower bound for the error that can be expected from the circuit execution.

\section*{Acknowledgment}
\small{This work was supported in part by the US Department of Energy, Office of Science, Early Career Award Program and by the U.S. Department of Energy, Office of Science, National Quantum Information Science Research Centers, Quantum Science Center. This research used resources of the Oak Ridge Leadership Computing Facility, which is a DOE Office of Science User Facility supported under Contract DE-AC05-00OR22725. This manuscript has been authored by UT-Battelle, LLC under Contract No. DE-AC05-00OR22725 with the U.S. Department of Energy. The United States Government retains and the publisher, by accepting the article for publication, acknowledges that the United States Government retains a non-exclusive, paid-up, irrevocable, worldwide license to publish or reproduce the published form of this manuscript, or allow others to do so, for United States Government purposes. The Department of Energy will provide public access to these results of federally sponsored research in accordance with the DOE Public Access Plan (http://energy.gov/downloads/doe-public-279access-plan).}

\bibliographystyle{ieeetr}
\bibliography{stability-references-aug2-2021.bib}

\section*{Appendix}
\subsection{Samples required for reliable estimation of $n$-dimensional distribution}\label{sec:appendix_precision}
To see this, consider an $n$ qubit register, which can have $N=2^n$ outcomes. We want to estimate the histogram for the $N$ outcomes. Suppose the true probability for the histogram bin $s$ (where $s \in \{0,1,\ldots, N-1\}$ is given by $p_s$ with:
\begin{equation}
\sum\limits_{s=0}^{N-1} p_s = 1.
\end{equation}
Suppose we run the experiment $T$ times and collect the $T$ $n$-bit strings. We are asking what is the minimum value for $T$ to achieve a desired precision. Let $Y_s$ denote the indicator variable, which is $1$ if $s$ is observed and $0$ otherwise. Thus, $Y_s$ is a Bernoulli variable, with mean $p_s$ and variance $p_s(1-p_s)$. After the $T$ circuit runs, the sample estimate for $p_s$ is:
\begin{equation}
\hat{p}_s = \frac{\sum\limits_{t=0}^{T-1} Y_s(t)}{T}
\end{equation}
which is the population mean for the time-series of the indicator variable $Y_s(t)$. Now suppose the desired precision ($\epsilon$), at a confidence level $\alpha$, is expressed as a percentage (for all s), as follows:
\begin{equation}
\textrm{Pr}\left( p_s(1-\epsilon) \leq \hat{p}_s \leq p_s(1+\epsilon) \right) = 1-\alpha
\label{eq:precision}
\end{equation}
\begin{equation}
\Rightarrow \textrm{Pr}\left( 
-\frac{p_s\epsilon}{\sqrt{p_s(1-p_s)/T}} \leq 
\frac{\hat{p}_s-p_s}{\sqrt{p_s(1-p_s)/T}} \leq 
\frac{p_s\epsilon}{\sqrt{p_s(1-p_s)/T}} \right) = 1-\alpha.
\end{equation}
However, from the central limit theorem, we know that $\frac{\hat{p}_s-p_s}{\sqrt{p_s(1-p_s)/T}}$ follows the standard normal distribution $z \sim \mathcal{N}(0,1)$. Now, for the standard normal variable z, we know that:
\begin{equation}
Pr\left( -z_{\alpha/2} \leq z \leq z_{\alpha/2} \right) = 1-\alpha
\label{eq:gaussian_alpha}
\end{equation}
where $z_{\alpha/2}$ is a constant corresponding to the two-sided confidence interval. Hence, if we set: 
\begin{equation}
\frac{\hat{p}_s-p_s}{\sqrt{p_s(1-p_s)/T}} = z_{\alpha/2}
\end{equation}
then both \mbox{Equation (\ref{eq:precision})} and \mbox{Equation (\ref{eq:gaussian_alpha})} are satisfied. Solving, we get the bound for T as:
\begin{equation}
\begin{split}
T &= \left( \frac{1}{p_s}-1\right) \frac{z_{\alpha/2}^2}{\epsilon^2}\\
&\sim \mathcal{O}\left( \frac{2^n z_{\alpha/2}^2}{\epsilon^2} \right)\\
\end{split}
\end{equation}
since $p_s \sim \mathcal{O}(2^{-n})$ for an n qubit register. Hence, the number of samples required to estimate an arbitrary distribution over $N=2^n$ outcomes with precision ($\epsilon$) scales exponentially with the Hilbert space dimension $n$ and as inverse square of the required precision $\epsilon$. 
\color{black}

\subsection{ }\label{sec:appendix}
\begin{lemma}
For the 1 qubit circuit, shown in \mbox{Figure \ref{fig:27_hadamards_with_n_1}}, 
if $\delta \leq 0.54$, then:
\begin{equation}
\gamma_D(\mathcal{C}) \leq \gamma_{max}(\mathcal{C}) \Leftrightarrow d \leq \delta
\label{eq:lemma}
\end{equation}
\end{lemma}

\begin{proof}[Proof.]
(i)
\begin{equation*}
\begin{split}
d \leq& \delta\\
\Rightarrow \sqrt{1-\frac{\sqrt{1+\gamma_D(\mathcal{C})}+\sqrt{1-\gamma_D(\mathcal{C})}}{2}} \leq \delta\\
\Rightarrow \sqrt{1+\gamma_D(\mathcal{C})}+\sqrt{1-\gamma_D(\mathcal{C})} \geq 2(1-\delta^2)\\
\Rightarrow \left(\sqrt{1+\gamma_D(\mathcal{C})}+\sqrt{1-\gamma_D(\mathcal{C})}\right)^2 \geq 4(1-\delta^2)^2\\
\Rightarrow \sqrt{1-\gamma_D(\mathcal{C})^2} \geq 2(1-\delta^2)^2-1\\.
\end{split}
\end{equation*}
Since $\delta \leq \sqrt{1-\frac{1}{\sqrt{2}}}$, the RHS is positive and we can square both sides to get:
\begin{equation*}
\begin{split}
\Rightarrow 1-\gamma_D(\mathcal{C})^2 \geq \left( 2(1-\delta^2)^2-1 \right)^2\\
\Rightarrow \gamma_D(\mathcal{C}) \geq 2(1-\delta^2)\sqrt{1-(1-\delta^2)^2}\\
\Rightarrow \gamma_D(\mathcal{C}) \leq \gamma_{max}(\mathcal{C}).
\end{split}
\end{equation*}
(ii) Now suppose: 
\begin{equation*}
\begin{split}
\gamma_D(\mathcal{C}) \leq&  \gamma_{max}(\mathcal{C})\\
\Rightarrow \gamma_D(\mathcal{C}) \geq 2(1-\delta^2)\sqrt{1-(1-\delta^2)^2}\\
\Rightarrow 1-\gamma_D(\mathcal{C})^2 \geq \left( 2(1-\delta^2)^2-1 \right)^2\\
\Rightarrow \sqrt{1-\gamma_D(\mathcal{C})^2} \geq 2(1-\delta^2)^2-1\\
\Rightarrow \left(\sqrt{1+\gamma_D(\mathcal{C})}+\sqrt{1-\gamma_D(\mathcal{C})}\right)^2 \geq 4(1-\delta^2)^2\\
\Rightarrow \sqrt{1+\gamma_D(\mathcal{C})}+\sqrt{1-\gamma_D(\mathcal{C})} \geq 2(1-\delta^2)\\
\Rightarrow \sqrt{1-\frac{\sqrt{1+\gamma_D(\mathcal{C})}+\sqrt{1-\gamma_D(\mathcal{C})}}{2}} \leq \delta\\
\Rightarrow d \leq& \delta.
\end{split}
\end{equation*}
\end{proof}

\begin{figure}[H]
\center
\includegraphics[width=8cm]{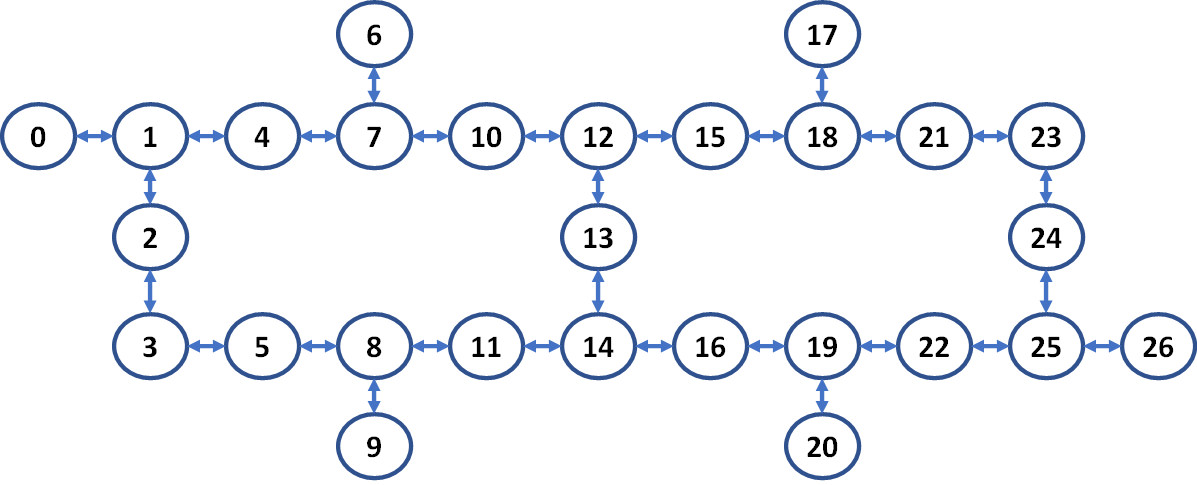}
\caption{Schematic of the \toronto device, produced by IBM. Circles represent register elements, while edges denote the connectivity for performing two qubit operations.}
\label{fig:toronto}
\end{figure}

\begin{figure}[H]
\centering
\includegraphics[width=8cm]{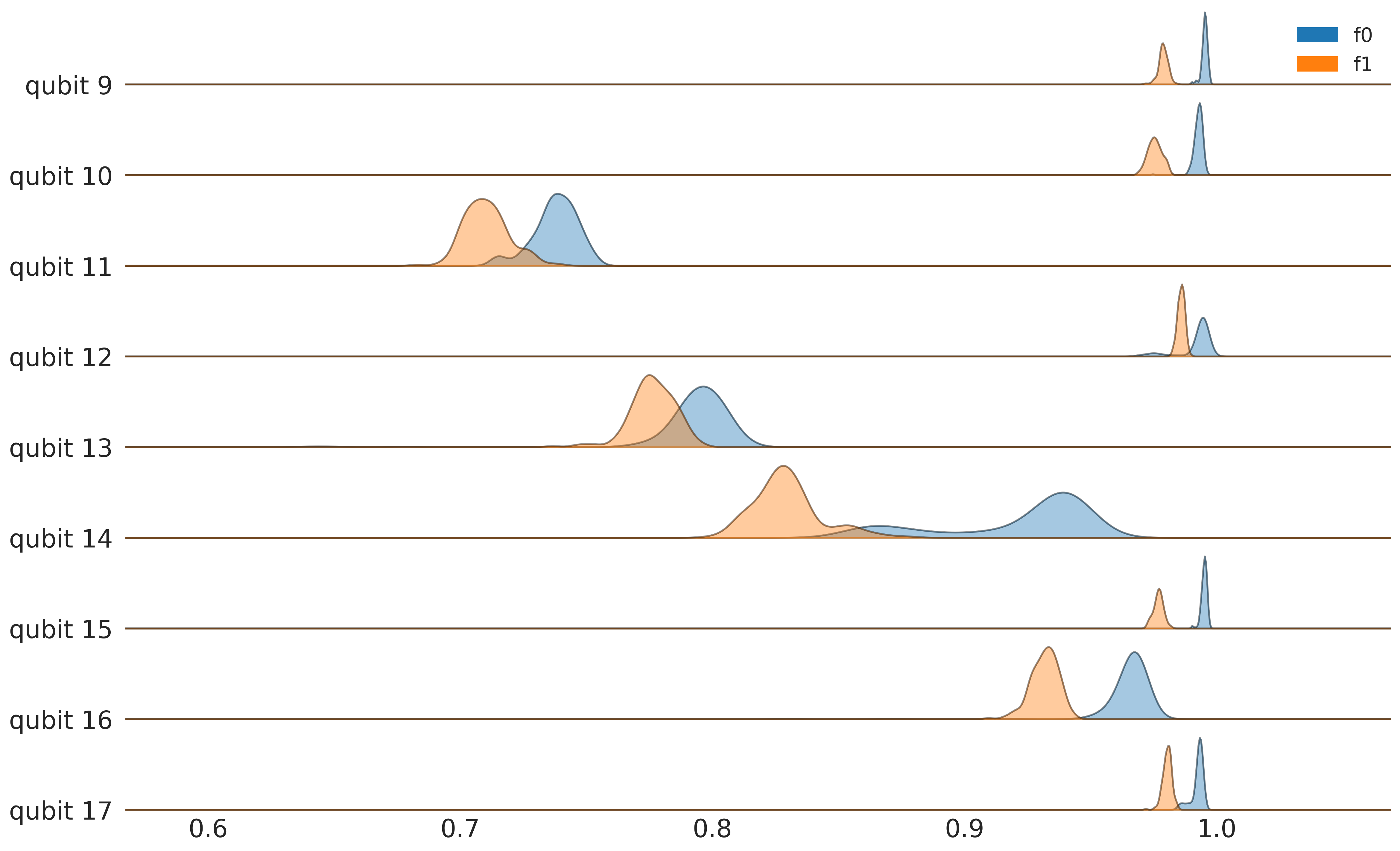}%
\caption{Fidelity distributions for computational basis states of \toronto for qubits $9-17$. Raw data for \toronto, collected on 8 April 2021, between 8:00-10:00pm (UTC-05:00).}
\label{fig:f0f1_toronto_qubit_9_onwards_spruce_2021}
\end{figure}

\begin{figure}[H]
\centering
\includegraphics[width=8cm]{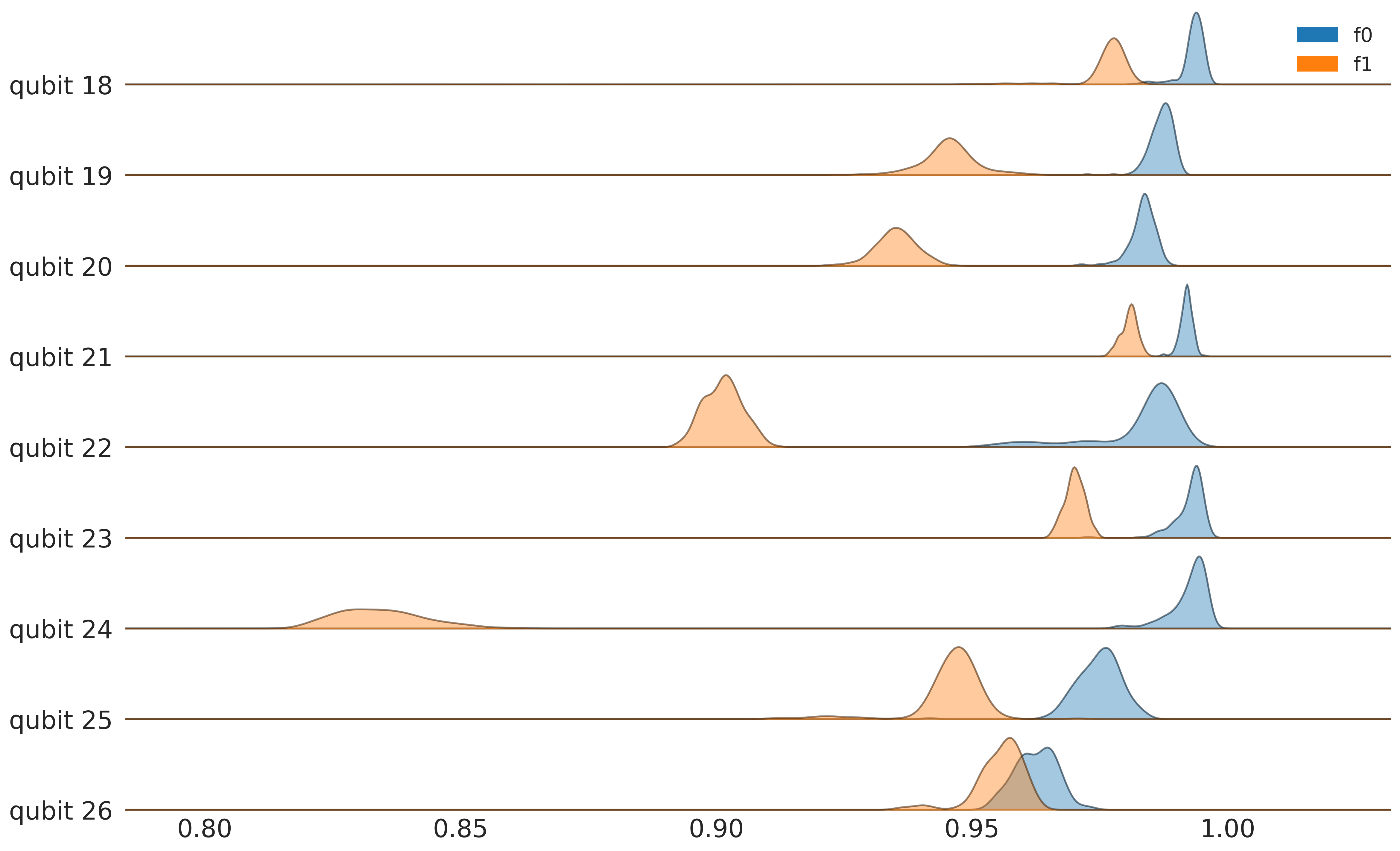}%
\caption{Fidelity distributions for computational basis states of \toronto for qubits $18-26$. Raw data for \toronto, collected on 8 April 2021, between 8:00-10:00pm (UTC-05:00).}
\label{fig:f0f1_toronto_qubit_18_onwards_spruce_2021}
\end{figure}

\begin{figure}[H]
  \centering
  \includegraphics[width=8cm]{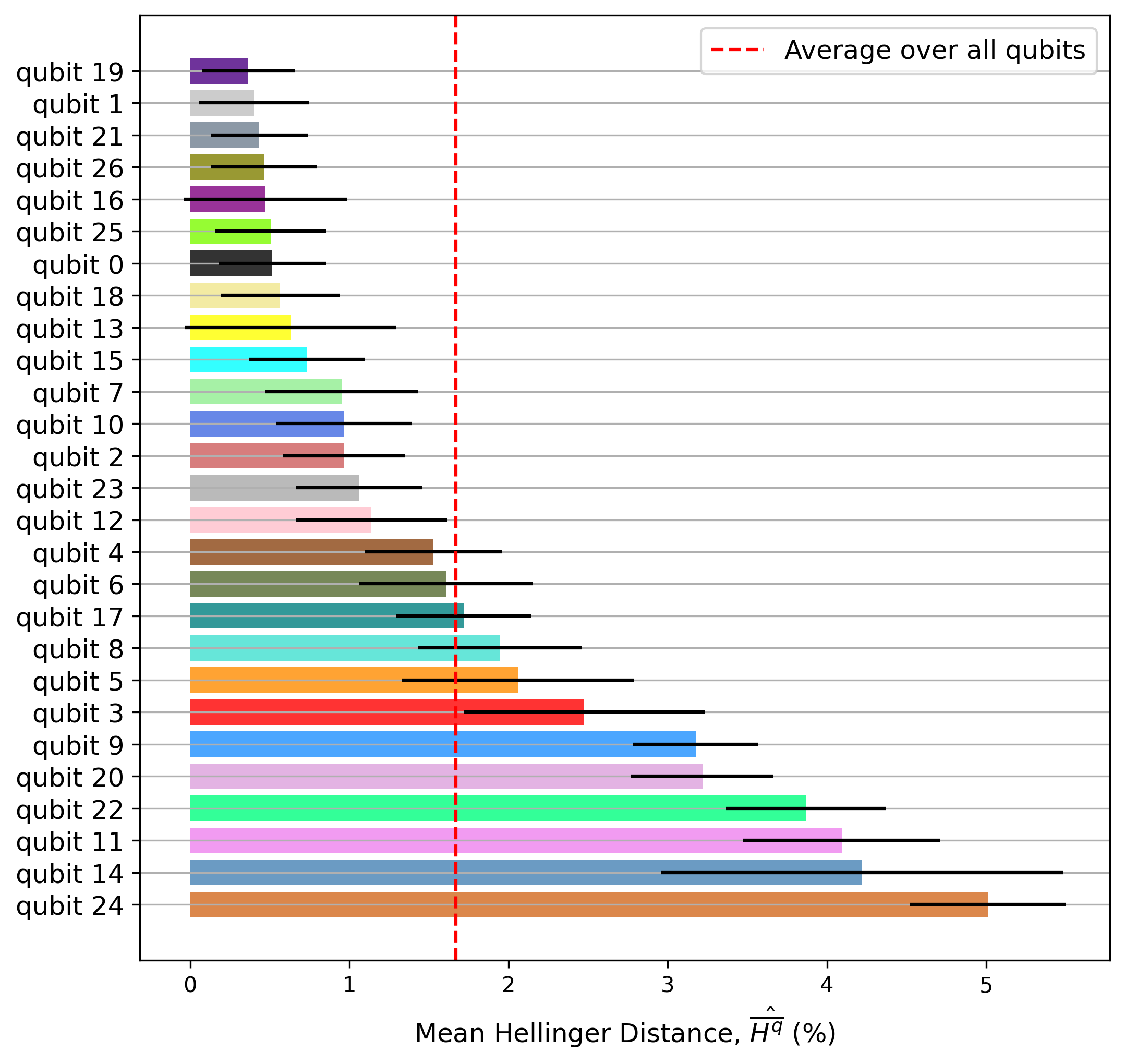}
  \captionof{figure}{Register variation of the experimentally-obtained Hellinger distance. Each of the 27 register elements was used to verify 
   \mbox{Equation (\ref{eq:repeated_ineq})} for $n=1$. The dotted red line denotes the register mean for Hellinger distance (averaged over all qubits). Qubit $19$ is the closest to ideal, while qubit $24$ is the farthest. The error bars show the standard deviation of the population mean across $L=203$ experiments. A consistent register color scheme has been used for all the figures.}
  \label{fig:hellinger_toronto_04082021_spruce.png}
\end{figure}%

\begin{figure}[H]
\center
\includegraphics[width=8cm]{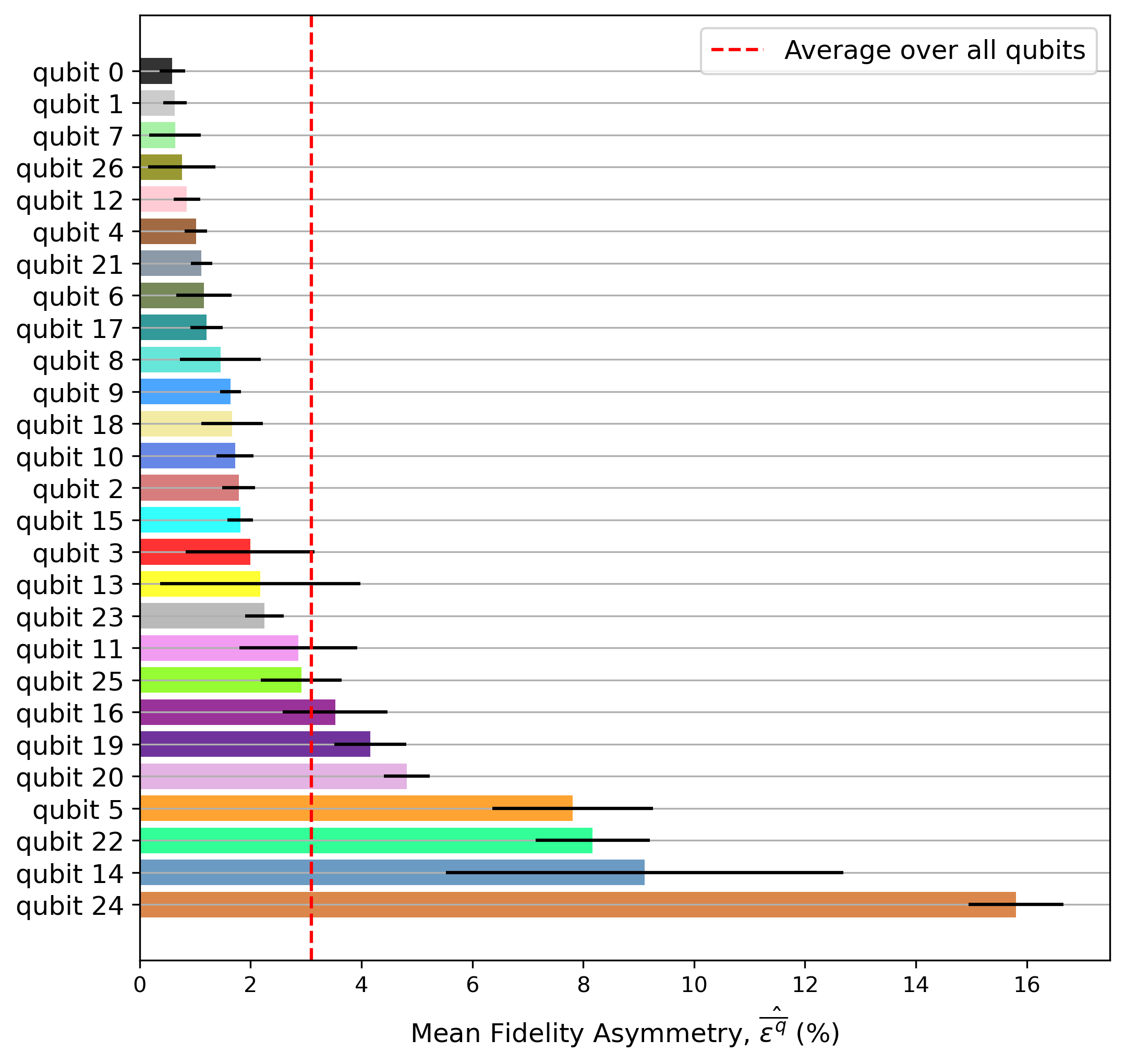}
\caption{Each of the 27 register elements was used to verify \mbox{Equation (\ref{eq:repeated_ineq})} for $n=1$. The dotted red line denotes the register mean for readout asymmetry (averaged over all qubits). Qubit $0$ has the best performance for this parameter, while qubit $24$ was the worst. The error bars show the standard deviation of the population mean across L=203 experiments. A consistent register color scheme has been used for all the figures.}
\label{fig:asymmetry_toronto_04082021_spruce}
\end{figure}

\begin{figure}[H]
  \centering
  \includegraphics[width=8cm]{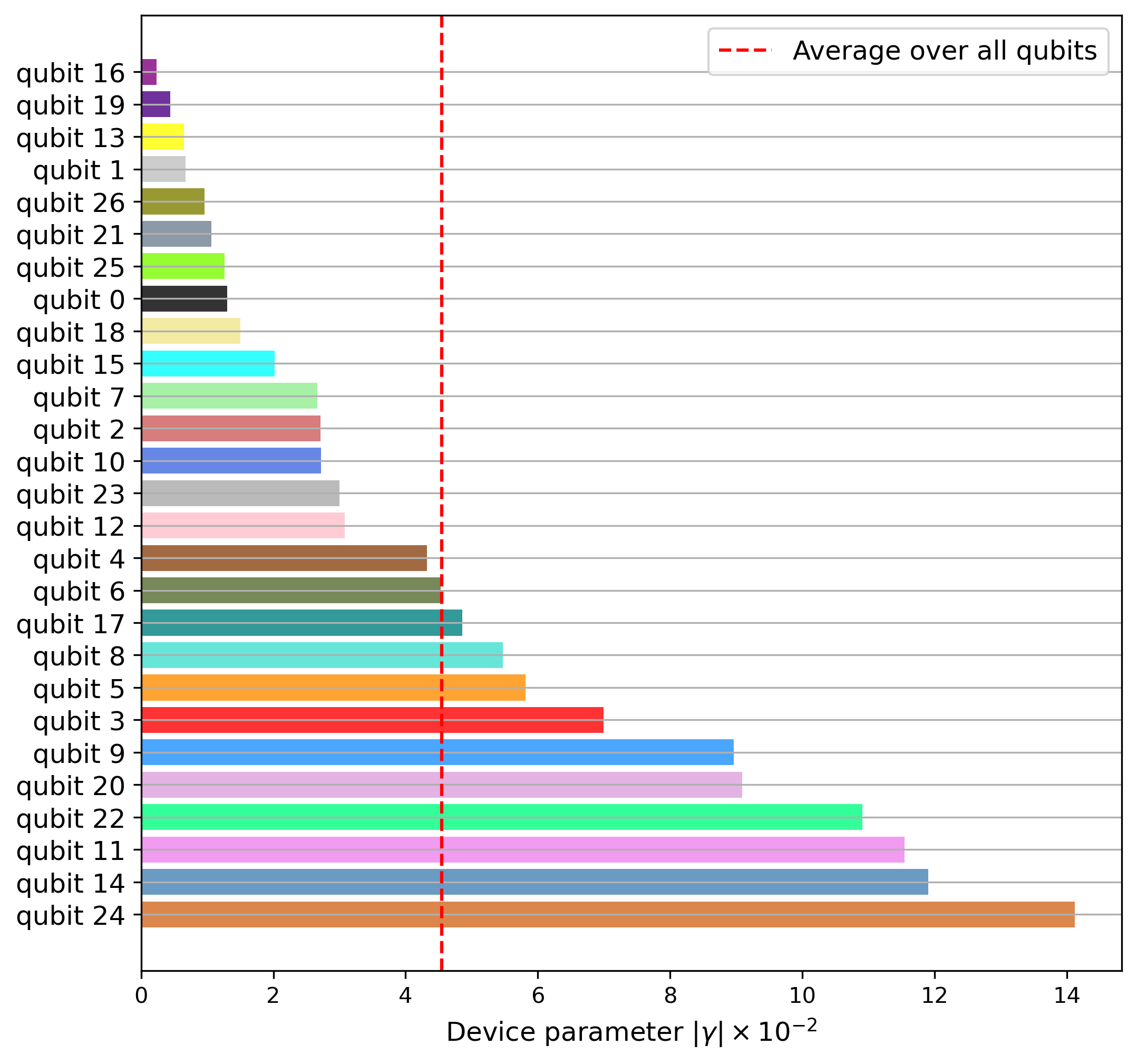}
  \captionof{figure}{Register variation of the derived parameter 
  $\gamma_D(\mathcal{C})$. A high $\gamma_D(\mathcal{C})$ adversely impacts the computational reproducibility. However, low $\gamma_D(\mathcal{C})$ does not necessarily mean that the device is close to perfect because the terms in \mbox{Equation (\ref{eq:theta_def})} can cancel each other out and lead to an improvement in the output distance. For an unstable device, $\gamma_D(\mathcal{C})$ will vary with time and must be re-estimated. The dotted red line denotes the register mean averaged over all qubits. Qubit $16$ has the best performance for this composite parameter, while qubit $24$ has the worst. A consistent register color scheme has been used for all the figures.}
  \label{fig:gamma_qubitwise}
\end{figure}%

\begin{table}[H]
\small
\caption{Register values for $\gamma_D(\tau)$ and $\gamma_{\textrm{max}}$. \label{tab:gamma_vals}}
\begin{tabular}{ccc}
\toprule
\textbf{Register \#}	& \textbf{$\gamma_{\textrm{max}}$} & \textbf{$\gamma_D(\tau)$}\\
\midrule
0&      1.4590&     1.3040\\
1&      1.1365&     0.6755\\
2&      2.7284&     2.7118\\
3&      6.9946&     6.9931\\
4&      4.3229&     4.3226\\
5&      5.8171&     5.8157\\
6&      4.5425&     4.5325\\
7&      2.6946&     2.6649\\
8&      8066&     5.4724\\
9&      8.9672&     8.9666\\
10&     2.7272&     2.7231\\
11&    11.5502&    11.5486\\
12&     3.2212&     3.0797\\
13&     1.7818&     0.6460\\
\bottomrule
\end{tabular}\\

\begin{tabular}{ccc}
\toprule
\textbf{Register \#}	& \textbf{$\gamma_{\textrm{max}}$} & \textbf{$\gamma_D(\tau)$}\\
\midrule
14&    11.9104&    11.9038\\
15&     2.0713&     2.0228\\
16&     1.3392&     0.2359\\
17&     4.8557&     4.8553\\
18&     1.5986&     1.4980\\
19&     1.0322&     0.4378\\
20&     9.0893&     9.0886\\
21&     1.2259&     1.0620\\
22&    10.9146&    10.9136\\
23&     3.0018&     3.0017\\
24&    14.1254&    14.1241\\
25&     1.4325&     1.2624\\
26&     1.3103&     0.9567\\
  &           &           \\
\bottomrule
\end{tabular}
\end{table}

\begin{figure}[H]
\center
\includegraphics[width=8cm]{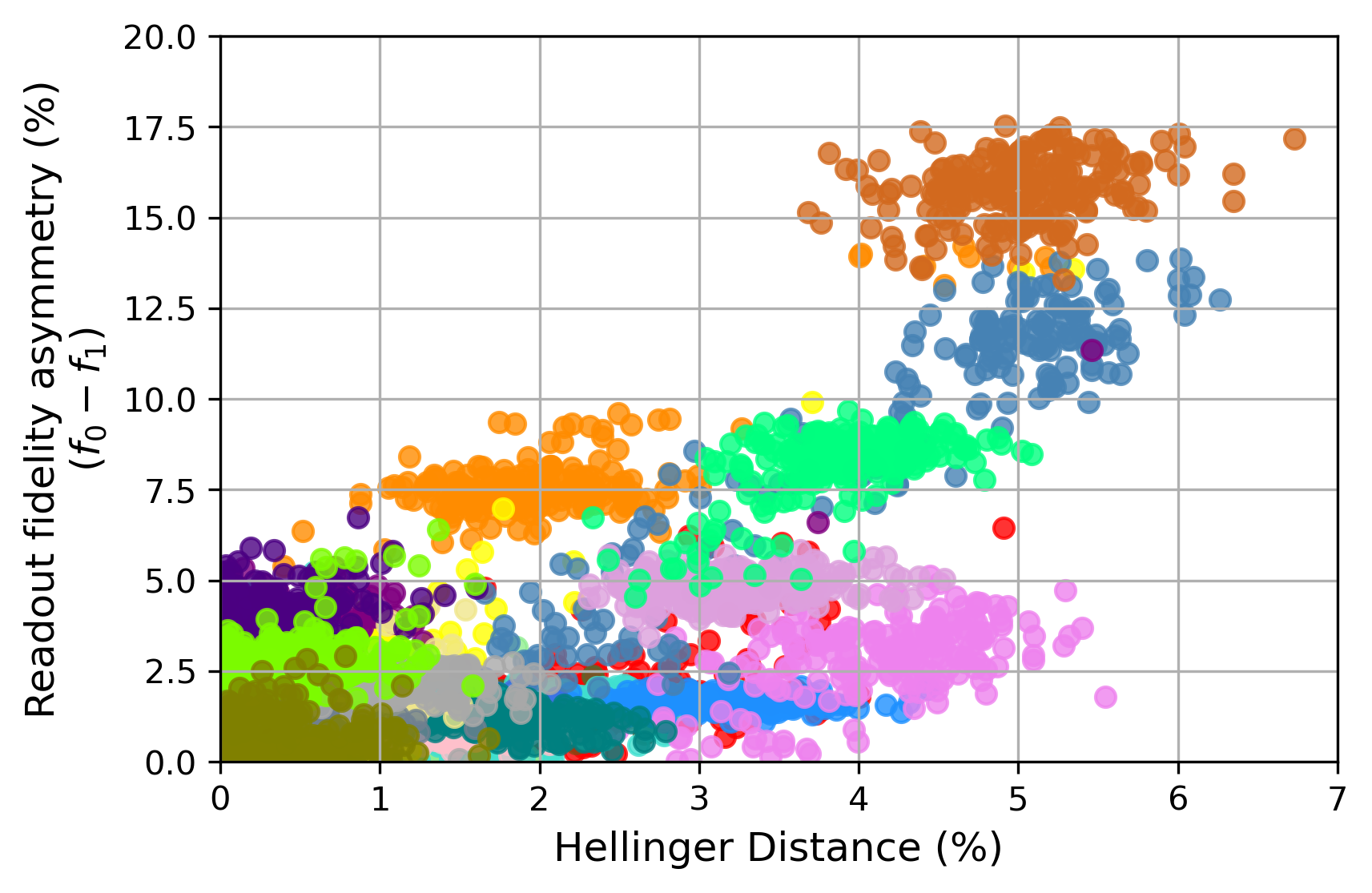}%
\caption{Plot of readout asymmetry vs observed Hellinger distance for 
$L=203$ experiments. Each register element is shaded by a different color.}
\label{fig:scatter_data}
\end{figure}
\end{document}